\documentclass[11pt,reqno]{article}   
\usepackage{fullpage}

\usepackage[utf8]{inputenc}
\usepackage[T1]{fontenc}
\usepackage{amsmath}
\usepackage{amssymb}
\usepackage{graphicx}
\usepackage[left=3.00cm, right=3.00cm]{geometry}

\usepackage{amsthm}

\theoremstyle{theorem}

\newtheorem{thm}{Theorem}
\newtheorem{prop}[thm]{Proposition}

\newtheorem{cor}[thm]{Corollary}

\theoremstyle{definition}
\newtheorem{defn}[thm]{Definition}
\newtheorem{rmk}[thm]{Remark}
\newtheorem{eg}[thm]{Example}

\newtheorem*{thm*}{Theorem}

\usepackage{tikz}
\usepackage{tikz-cd}
\usetikzlibrary{quantikz}
\usetikzlibrary{decorations.pathreplacing,decorations.markings,angles,quotes}

\usepackage{euscript}
\usepackage{hyperref}
\usepackage{bbold}
\def\on{\operatorname}

\def\bb{\mathbb}
\def\CC{\bb{C}}
\def\RR{\bb{R}}

\def\ZZ{\bb{Z}}
\def\CP{\on{CP}}

\def\lL{\mathcal{L}}

\def\hH{\mathcal{H}}
\def\kK{\mathcal{K}}

\newcommand{\Span}[1]{\ensuremath{ \langle #1 \rangle }}
\newcommand{\one}{\mathbb{1}}  
\newcommand{\zero}{\mathbb{0}}  
\newcommand{\set}[1]{\ensuremath{ \lbrace #1 \rbrace }}

\def\myvdots{\ \vdots\ }

\newcommand{\Herm}{\on{Herm}}  
\newcommand{\Tr}{\on{Tr}}  
\newcommand{\Conv}{\on{Conv}}

\usetikzlibrary{calc}

\newcommand{\gettikzxy}[3]{%
	\tikz@scan@one@point\pgfutil@firstofone#1\relax
	\edef#2{\the\pgf@x}%
	\edef#3{\the\pgf@y}%
}
\makeatother

\newcommand{\Block}[4]{
	\node (#4) at (#1,#2) [draw,thick, minimum width=3em,minimum height=3em] {#3};}

\newcommand{\VOut}[3]{
	($(#1.center)+({-3+((6*#3)/(#2+1))},-3)$)
}
\newcommand{\VIn}[3]{
	($(#1.center)+({-3+((6*#3)/(#2+1))},3)$)
}
\newcommand{\HIn}[3]{
	($(#1.center)+(-3,{3-((6*#3)/(#2+1))})$)
}
\newcommand{\HOut}[3]{
	($(#1.center)+(3,{3-((6*#3)/(#2+1))})$)
}

\newcommand{\HDraw}[8]{
	\draw let 
	\p1=\HOut{#1}{#2}{#3},
	\p2 =\HIn{#4}{#5}{#6}
	in 
	(\x1,\y1) to[out=#7,in=#8] (\x2,\y2);	
}

\newcommand{\VDraw}[8]{
	\draw[dashed] let 
	\p1=\VOut{#1}{#2}{#3},
	\p2 =\VIn{#4}{#5}{#6}
	in 
	(\x1,\y1) to[out=#7,in=#8] (\x2,\y2);	
}

\newcommand{\VIndraw}[4]{
	\draw[dashed] 
	let 
	\p1=\VIn{#1}{#2}{#3},
	in 
	(\x1,\y1) to ($(\x1,\y1)+(0,#4)$);}

\newcommand{\VOutdraw}[4]{
	\draw[dashed] 
	let 
	\p1=\VOut{#1}{#2}{#3},
	in 
	(\x1,\y1) to ($(\x1,\y1)+(0,-#4)$);}

\newcommand{\HIndraw}[4]{
	\draw
	let 
	\p1=\HIn{#1}{#2}{#3},
	in 
	(\x1,\y1) to ($(\x1,\y1)+(-#4,0)$);}
\newcommand{\HOutdraw}[4]{
	\draw
	let 
	\p1=\HOut{#1}{#2}{#3},
	in 
	(\x1,\y1) to ($(\x1,\y1)+(#4,0)$);}

\newcommand{\FCircle}[3]{%
  \node[circle,draw,fill=black,minimum size=0.5em,inner sep=0pt,transform shape]
    (#3) at (#1,#2) {};%
}

\newcommand{\CVIn}[3]{(#1.north)}  
\newcommand{\CVOut}[3]{(#1.south)} 

\newcommand{\VDrawBC}[8]{%
  \draw[dashed] let \p1=\VOut{#1}{#2}{#3}, \p2=\CVIn{#4}{#5}{#6}
  in (\x1,\y1) to[out=#7,in=#8] (\x2,\y2);%
}
\newcommand{\VDrawCB}[8]{%
  \draw[dashed] let \p1=\CVOut{#1}{#2}{#3}, \p2=\VIn{#4}{#5}{#6}
  in (\x1,\y1) to[out=#7,in=#8] (\x2,\y2);%
}

\newcounter{commcount}
\usepackage{xcolor}
\usepackage{todonotes}

 \usepackage{authblk}

\newcommand{{\LocP}}{\textsc{LocalPauli}}
\newcommand{{\Lam}}{\textsc{Lambda}}
\newcommand{{\CNC}}{\textsc{Cnc}}
\newcommand{{\JW}}{\textsc{JordanWigner}} 
\newcommand{{\LocC}}{\textsc{LocalClosed}} 
\newcommand{{\Det}}{\textsc{Deterministic}} 
\newcommand{{\MaxW}}{\textsc{MaxWeight}} 
\newcommand{{\Wig}}{\textsc{Wigner}} 
\newcommand{{\Stab}}{\textsc{Stabilizer}}

\usetikzlibrary{arrows}        
\usetikzlibrary{arrows.meta}   
 
\tikzset{
  mid reverse triangle/.style={
    draw=black,
    postaction={
      decorate,
      decoration={
        markings,
        mark=at position 0.5 with {\arrow[rotate=180]{triangle 45}}
      }
    } 
  }
}

\tikzset{
  mid triangle/.style={
    draw=black,
    postaction={
      decorate,
      decoration={
        markings,
        mark=at position 0.5 with {\arrow{triangle 45}}
      }
    } 
  }
}

\title{Polyhedral Classical Simulators for Quantum Computation} 
\begin{document}
 
\author{Cihan Okay\footnote{cihan.okay@bilkent.edu.tr}} 
\affil{{\small{Department of Mathematics, Bilkent University, Ankara, Turkey}}}

	\maketitle

\begin{abstract}
Quantum advantage in computation refers to the existence of computational tasks 
that can be performed efficiently on a quantum computer but cannot be efficiently 
simulated on any classical computer. Identifying the precise boundary of efficient classical simulability is a central challenge and motivates the development of new simulation paradigms.  
In this paper, we introduce {polyhedral classical simulators}, a framework for classical simulation grounded in polyhedral geometry.
This framework encompasses well-known methods such as the Gottesman--Knill algorithm, while also extending naturally to more recent models of quantum computation, including those based on magic states and measurement-based quantum computation.  
We show how this framework unifies and extends existing simulation methods while at the same time providing a geometric roadmap for pushing the boundary of efficient classical simulation further. 
\end{abstract}

\tableofcontents

\section{Introduction}
\label{sec:intro}

A quantum computer operates using the principles of quantum mechanics. These operations provide new computational phenomena not available in the classical realm. Understanding exactly which features of quantum theory yield such phenomena---in particular, speedup in computational tasks---is marked as a quantum advantage. Such an advantage usually comes from foundational aspects of quantum theory, such as non-locality and entanglement, hence providing a separation from the classical world, and manifests itself in quantum algorithms that cannot be efficiently simulated by any classical computer. Our main focus in this paper is quantum advantage in computation, that is, when efficient classical simulation fails. As we will see, foundational features of quantum theory enter naturally as first indicators of potential quantum advantage.

In this paper, we introduce a class of classical simulators for quantum computation, providing a rigorous geometric framework based on polyhedral methods. The basic idea is to construct a classical state space consisting of the vertices---that is, the extremal points---of a convex polytope, which we refer to as the \emph{simulation polytope}. The classical algorithm initiates by sampling from this state space and proceeds by probabilistically updating these classical states. In this way, the algorithm can reproduce the Born rule probabilities obtained from the quantum algorithm. The choice of polytope depends on the computational model used to implement the quantum algorithm. Then, the simulation polytope is defined in a dual fashion to the measurement operations that appear in the computational model. Geometrically, two cases can arise: (1) the simulation polytope does not contain the convex set of all quantum states, also known as density operators, or (2) the polytope contains all quantum states. The former does not yield an initial probability distribution if the quantum computation starts at a quantum state outside the polytope. In this case, the initial distribution is a quasi-probability distribution, allowing some negative values. Depending on these two situations, we can distinguish two broad classes of simulators (see Figure \ref{fig:poly}):  
\begin{itemize}
\item \textbf{Efficient Estimators:} A quasi-probabilistic simulator that \emph{estimates} Born rule probabilities and performs efficient sampling when the initial quantum state lies within the simulation polytope.  
\item \textbf{Universal Samplers:} A probabilistic simulator that \emph{samples} from Born rule distributions, with a simulation polytope that contains all quantum states.  
\end{itemize}
The geometric restriction in the first case imposes limits on the class of quantum circuits that can be handled by the resulting polyhedral simulator.

\begin{figure}[ht]  
\centering 
\begin{tikzpicture}[scale=2, every node/.style={circle,inner sep=1pt}]

\def\n{7}
\def\radius{1}
\def\inradius{0.8}

\foreach \i in {1,...,8} {
  \node[fill=red]  (A\i) at ({\radius*cos(360/8*(\i-1))}, {\radius*sin(360/8*(\i-1))}) {};
}

\foreach \i in {1,...,7} {
  \pgfmathtruncatemacro{\next}{\i+1}
  \draw[gray!70, thick] (A\i) -- (A\next);
}
\draw[gray!70, thick] (A8) -- (A1);

\draw[black, thick,fill=gray!20] (0,0) circle (\inradius);

\foreach \i/\j in {1/3, 3/5, 5/7, 7/1} {
  \draw[black, thick] (A\i) -- (A\j);
}
\foreach \i in {1,3,5,7} {
  \fill[blue] (A\i) circle (1pt);
}

\foreach \i in {2,4,6,8} {
  \fill[red] (A\i) circle (1pt);
}

\draw[->,dashed, thick, blue!50!black, rounded corners=10pt]
     (A3) .. controls +(1,.5) and +(.5,.5) .. (A1);

\draw[->, dashed, thick, blue!50!black, rounded corners=10pt]
     (A3) .. controls +(-1,.6) and +(.8,.8) .. (A3);
 
\end{tikzpicture} 
\caption{\label{fig:poly}  
Geometry of quantum states and simulation polytopes. In the case of an efficient estimator (shown as a diamond with blue vertices), the polytope covers only part of the quantum state space (depicted as a disk). A universal sampler (depicted as the larger polytope), by contrast, contains the entire quantum state space. Red vertices indicate iterative progress toward probabilistically simulating larger regions. The update rules specify a probabilistic choice of the next vertex---for example, remaining at the same vertex with probability $p$, or moving to another vertex with probability $1-p$.
} 
\end{figure}
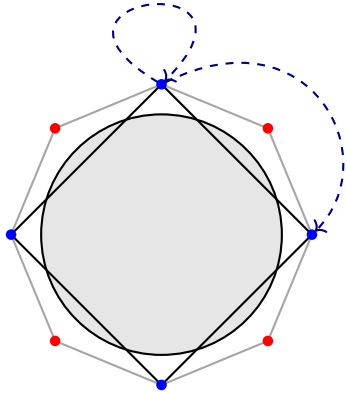 

The classic example of an efficient estimator is the stabilizer tableau simulator of Gottesman--Knill \cite{gottesman1999group22}, which we review in Section \ref{sec:quantum circuits}. The simulation polytope in this case is the convex hull of the stabilizer states, forming the \emph{stabilizer polytope} $\on{SP}_n$. These are special types of quantum states defined using a finite group, known as the Pauli group in the quantum computing literature. The algebraic nature of this theory makes the quantum operations highly controllable and tractable. This subtheory has many applications in quantum computing, including quantum error correction \cite{nielsen2010quantum}. A quantum circuit built within the stabilizer theory is called a \emph{stabilizer circuit}. The tableau algorithm of Gottesman--Knill can efficiently simulate any stabilizer circuit. It is well known that stabilizer circuits are not \emph{universal}, in the sense that not every quantum computation can be performed using them.  Bravyi--Kitaev \cite{bravyi2005universal} observed that stabilizer circuits can be extended to a universal model of quantum computation, known as quantum computation with magic states (QCM). In this model, if a stabilizer circuit is allowed to begin with a special type of quantum state called a \emph{magic state}---lying outside the stabilizer polytope---then the resulting class of circuits becomes universal. This naturally raises the question of whether there exists a polytope that contains both the stabilizer polytope and these magic states.

Several attempts have been made to enlarge the stabilizer polytope in order to capture more quantum states. The first approach, introduced in \cite{veitch2012negative}, replaces stabilizer states with operators originating from quantum optics.
 The corresponding simulation algorithm is based on the \emph{Wigner polytopes} $\on{WP}_n$. However, this method works only for qudits of odd local dimension and, in particular, fails for qubits. Qudits are described by a Hilbert space of the form $(\CC^d)^{\otimes n}$, where $d$ is referred to as the \emph{local dimension}, with the special case $d=2$ corresponding to qubits.  
A more general phase-space simulation method, based on \emph{closed non-contextual operators} (CNC), was later developed in \cite{raussendorf2020phase} for qubits, forming the \emph{CNC polytope} $\on{CP}_n$. Its qudit extension \cite{zurel2024efficient} provides a strict generalization of the Wigner simulation. In \cite{ipek2025phase}, the authors, including the present author, introduced a tableau method for phase-space simulation that represents CNC operators and their updates, extending the stabilizer tableau method. This algorithm has been implemented in \cite{CNCSim}.

The motivation for polyhedral simulators becomes apparent with two key examples. The first universal sampling simulator was introduced in \cite{zurel2020hidden} for qubits and later extended in \cite{zurel2021hidden} to qudits. The associated simulation polytope will be referred to as the \emph{Pauli polytope}, denoted $\on{P}_n$.%
\footnote{In \cite{zurel2020hidden}, the notation $\Lambda_n$ was used, and the polytope was referred to as the “Lambda polytope.” We prefer the term Pauli polytope, as it highlights the type of measurements that define the polytope. Moreover, in its local version the polytope is called the {local Pauli polytope} \cite{okay2024classical}.} 
It contains all quantum states, and the computational model here is QCM, where the measurements are Pauli measurements.
A more recent development is the universal sampler of \cite{okay2024classical}, whose simulation polytope is called the \emph{local Pauli polytope}. In this model, computation proceeds with Pauli measurements acting on a single qubit, making it a local variant of QCM. More precisely, this construction realizes an instance of measurement-based quantum computation (MBQC), another fundamental scheme originally introduced by Raussendorf--Briegel \cite{raussendorf2001one}.

In Section \ref{sec:adaptive quantum computation}, we introduce a formulation of adaptive quantum computation using \emph{adaptive instruments}. In quantum theory, instruments provide a formalization of quantum operations \cite{watrous2018theory}, including unitary transformations and quantum measurements. We further allow for adaptivity, meaning that the choice of an instrument at a later stage may depend on the outcome of an earlier instrument. This framework provides a uniform treatment of the circuit model as well as the two other universal models, QCM and MBQC. In particular, we focus on the Pauli and local Pauli models, which serve as the computational models underlying the definitions of the Pauli and local Pauli polytopes.  
In Section \ref{sec:polyhedral classical simulation}, we introduce the notion of \emph{polyhedral classical simulation}, which encapsulates the constructions shown in Figure \ref{fig:polytope-hierarchy}. As a visual aid, it is helpful to depict adaptive instruments and the components of classical simulation as diagrams consisting of boxes with vertical/horizontal inputs and vertical/horizontal outputs. The horizontal and vertical directions track the quantum and classical components of the adaptive computation, respectively. This double input/output structure can be formalized using \emph{double categories} \cite{double}.

The goal of the polyhedral classical simulators research program is to push the boundaries of probabilistic simulation further, by exploring new vertices of universal samplers (indicated as red vertices in Figure \ref{fig:poly}). For example, starting from $\on{CP}_n$, one may ask: which quantum states lie inside this polytope? The phase space tableau algorithm of \cite{ipek2025phase} provides an efficient simulation for the corresponding circuits. Two facts follow immediately from the geometry: (1) the polytope contains quantum states that are not stabilizer states, and (2) there are quantum states that do not lie within this polytope. Understanding the first extends the class of efficiently simulatable circuits beyond stabilizer circuits, while the second gives a clear boundary for quantum advantage, provided by circuits that fall outside this class.  
There may still exist quantum circuits that are efficiently simulatable beyond this particular class, so the boundary of efficient simulatability can in principle be pushed further. Indeed, progress in this direction is illustrated in Figure \ref{fig:polytope-hierarchy}: the \emph{Jordan--Wigner polytope} $\on{JW}_n$, consisting of line graph operators \cite{zurel2023simulation}, provides new vertices of the Pauli polytope, while on the local Pauli side the \emph{deterministic polytope} $\on{DP}_n$ and the \emph{max-weight polytope} $\on{MP}_n$ introduce vertices described using quantum contextuality and non-locality \cite{okay2024classical}.

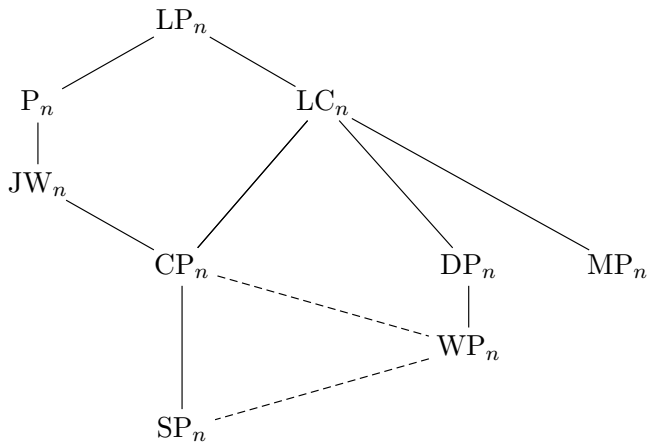
\begin{figure}[ht]
    \centering
{
\tikzcdset{every arrow/.append style={shorten <=3pt,shorten >=3pt}}
\begin{tikzcd}[row sep=2em, column sep=3em]
 & \on{LP}_n \arrow[dl,no head] \arrow[dr,no head] & &\\
\on{P}_n \arrow[d,no head] & & \on{LC}_n \arrow[ddl,no head] \arrow[ddr,no head] \arrow[ddl,no head] \arrow[ddrr,no head] &\\
\on{JW}_n \arrow[dr,no head] & && &\\
 & \on{CP}_n \arrow[dd,no head] \arrow[drr,no head,dashed] & & \on{DP}_n \arrow[d,no head] & \on{MP}_n  \\
 &  & & \on{WP}_n \arrow[dll,no head,dashed]&\\
 & \on{SP}_n && &
\end{tikzcd}
}
 \caption{Containment relation of simulation polytopes. Dashed arrows only hold for qudits of odd local dimension.}
    \label{fig:polytope-hierarchy}
\end{figure}

On the quantum foundations side, the key notions of contextuality and non-locality are most naturally formulated within the framework of \emph{simplicial distributions}, introduced in \cite{sdist} using simplicial methods from algebraic topology. Both the Pauli polytope and the local Pauli polytope can be expressed as special cases of polytopes of simplicial distributions: the latter coincides with the well-known Bell scenarios \cite{okay2024classical}, while the former can be described as a polytope of twisted simplicial distributions \cite{okay2024twisted}. Techniques from this theory have proven especially effective in addressing the vertex-enumeration problem \cite{kharoof2023topological,kharoof2024homotopical,kharoof2024extremal}, which, as discussed above, is a fundamental challenge for the state spaces of these simulators. 
In this paper, we focus on the simulation aspects. At the same time, the natural appearance of foundational notions highlights a systematic connection to the theory of simplicial distributions. This connection forms a unified framework for studying the foundations of quantum advantage alongside polyhedral classical simulation.

As a reading guide, Section \ref{sec:quantum circuits} begins with a conventional introduction to quantum circuits and highlights how the stabilizer subtheory can be efficiently simulated using $\ZZ_2$-linear operations. Section \ref{sec:adaptive quantum computation} then develops a more general framework of quantum computation that encompasses various standard models. This section marks the starting point of our rigorous and systematic treatment of the subject, whereas the preceding section is intended to convey the idea of simulation in its simplest form. Finally, in Section \ref{sec:polyhedral classical simulation}, we introduce our polyhedral simulation framework, maintaining the same level of rigor. The simulators discussed here include several prominent methods from the quantum computing literature, such as the Gottesman--Knill simulation and its more recent extensions.

\paragraph{Acknowledgments.}
This work is supported by the Air Force Office of Scientific Research (AFOSR) under award number  FA9550-24-1-0257 and the Digital Horizon Europe project FoQaCiA, GA no. 101070558. The author thanks Selman Ipek for valuable feedback on an earlier draft.

\paragraph{Data availability.}
Data sharing is not applicable to this article, as no datasets were generated or analyzed during the study.
 
\section{Quantum computation}
\label{sec:quantum circuits}

Quantum theory is defined on a Hilbert space $\hH$, which in quantum computing 
is typically taken to be finite-dimensional. It consists of three essential 
components: \emph{states}, \emph{transformations}, and \emph{measurements}.  
These are described using linear operators $A:\hH\to \hH$. An operator is 
called \emph{positive semidefinite} if it can be expressed as $A=B^\dagger B$ 
for some operator $B$, where $B^\dagger$ denotes the adjoint of $B$. We write 
$\on{Pos}(\hH)$ for the set of positive semidefinite operators.
\begin{itemize}
\item {\bf States:} A quantum state is given by a density operator, i.e., a 
trace-one positive semidefinite operator. We write $\on{Den}(\hH)$ for the set 
of density operators. 
\item {\bf Transformations:} Transformations are given by unitary operators on 
$\hH$, which form the unitary group $U(\hH)$.
\item {\bf Measurements:} A (projective) quantum measurement with finite outcome 
set $\Sigma$ is a function $\Pi:\Sigma\to \on{Proj}(\hH)$ assigning to each 
$s\in\Sigma$ a projection operator $\Pi^s$ such that
\[
\sum_{s\in \Sigma} \Pi^s = \one,
\]
where $\one$ denotes the identity operator. 
\end{itemize}
{Transformations act by conjugation on states. The action of measurements on states is probabilistic:}
Given a state $\rho$ and a measurement $\Pi$, the Born rule provides a 
probability distribution
\[
p:\Sigma \to \RR_{\geq 0}
\]
specifying the probability of observing outcome $s\in\Sigma$:
\[
p^s = \Tr(\rho\,\Pi^s).
\]
After observing outcome $s$, the state updates to the \emph{post-measurement 
state}
\[
\rho' = \frac{\Pi^s \rho \Pi^s}{p^s}.
\]

Quantum computation is performed using qubits, the quantum analogue of bits 
$\ZZ_2=\set{0,1}$. The \emph{$n$-qubit Hilbert space} is the $n$-fold tensor 
product $(\CC^2)^{\otimes n}$. It is common practice to use Dirac notation, 
writing the canonical basis vectors as $\ket{s}$, where $s=s_1s_2\cdots s_n$ is 
a bit string. The adjoint of this vector is denoted by $\bra{s}$. 
Explicitly, the quantum circuit in Figure~\ref{fig:circuit} associates a 
probability distribution to each input state. For a set $X$, let $D(X)$ denote 
the set of probability distributions on $X$. For a given unitary $U$, this 
association gives a function
\[
p:\ZZ_2^n \to D(\ZZ_2^n)
\] 
defined by the Born rule
\begin{equation}\label{eq:Born}
p_{s}^r = \Tr(U \Pi^{s} U^\dagger \Pi^{r}),
\end{equation}
where $\Pi^s=\ket{s}\bra{s}$ and $\Pi^r=\ket{r}\bra{r}$ are the projectors onto 
the basis vectors $\ket{s}$ and $\ket{r}$, respectively. {For each $s\in \ZZ_2^n$, we denote by $p_s$ the corresponding distribution in $D(\ZZ_2^n)$, and by $p_s^r$ its value at $r\in \ZZ_2^n$.}
The post-measurement 
state when $r$ is observed is the projector $\Pi^r$.
We say that the quantum circuit computes a function $f:\ZZ_2^n\to \ZZ_2^n$ if 
$p$ factors as the composite
\[
p:\ZZ_2^n \xrightarrow{f} \ZZ_2^n \xrightarrow{\delta} D(\ZZ_2^n),
\]
where for a set $X$, the function $\delta:X\to D(X)$ sends each element to the 
delta distribution peaked at that element. In certain cases, it is possible to 
produce the distribution $p_s$ for each input by classical means efficiently. 
Such cases are regarded as giving no quantum advantage.

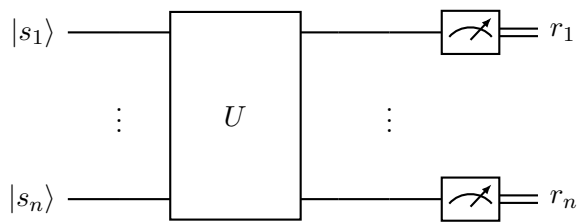
\begin{figure}[h]  
    \centering
    \begin{quantikz} 
        \lstick{$\ket{s_1}$} & \qw & \gate[3,nwires=2]{\hspace{1.5em} U \hspace{1.5em}} &  \qw &  \qw & \meter{} & \cw \rstick{$r_1$}\\
         &  \myvdots    &    &      &   \myvdots  \\
        \lstick{$\ket{s_n}$}  &\qw &   & \qw & \qw & \meter{} & \cw \rstick{$r_n$}
    \end{quantikz}
    \caption{Quantum circuit. The input state is the canonical basis vector 
$\ket{s_1\cdots s_n}=\ket{s_1}\otimes \cdots\otimes \ket{s_n}$ 
with corresponding projector $\Pi^s$. The unitary transforms this state to 
$U\ket{s_1\cdots s_n}$, corresponding to the projector $U\Pi^s U^\dagger$. 
Finally, a measurement is performed with projectors 
$\{\Pi^r : r=r_1\cdots r_n\}$.
}
    \label{fig:circuit}
\end{figure}

\subsection{Stabilizer theory} 
\label{sec:stabilizer theory}
 
Quantum computational power can be analyzed from a group-theoretic perspective. 
The Pauli matrices
\[ 
X = \begin{pmatrix} 0 & 1 \\ 1 & 0\end{pmatrix},\quad
Y = \begin{pmatrix} 0 & -i \\ i & 0\end{pmatrix},\quad
Z = \begin{pmatrix} 1 & 0 \\ 0 & -1\end{pmatrix}
\]
form the starting point for constructing an important finite subgroup. 
The \emph{$n$-qubit Pauli group} $G_n$ is the subgroup of 
$U((\CC^2)^{\otimes n})$ generated by tensor products 
$A_1\otimes A_2\otimes \cdots \otimes A_n$, where each $A_i$ is one of 
$X, Y, Z$, or the identity $\one$.

\begin{figure}[h] 
    \centering
\begin{tikzpicture}[baseline=(current bounding box.center)]

\node (G5) at (0,8) {\( U(\mathcal{H}) \)}; 
\node (G4) at (0,6) {\( \langle H_i, T_j, CZ_{kl} \rangle \)};
\node (G3) at (0,4) {\( \on{Cl}_n= \langle H_i, S_j, CZ_{kl} \rangle \)};  
\node (G2) at (0,2) {\( G_n= \langle A_1 \otimes A_2\otimes  \cdots \otimes A_n : A_i =\mathbb{1}, X, Y, Z  \rangle \)};  

\draw[double, double distance=1pt, line width=0.5pt] (0,2.2) -- (0,3.8);
\draw[thick] (0,4.2) -- (0,5.8);
\draw[thick] (0,6.2) -- (0,7.8);
\end{tikzpicture}
    \caption{Subgroups of the unitary group relevant to quantum computation. The first is a 
dense subgroup that arises in the proof of quantum universality. The second is 
the finite Clifford group $\on{Cl}_n$, which plays a central role in the 
construction of stabilizer (or Clifford) circuits. The last, $G_n$, is the Pauli 
group—an (almost) extraspecial $2$-group underlying the stabilizer subtheory of 
quantum mechanics. Double lines indicate that $G_n$ is normal in $\on{Cl}_n$.
}
    \label{fig:universality}
\end{figure}
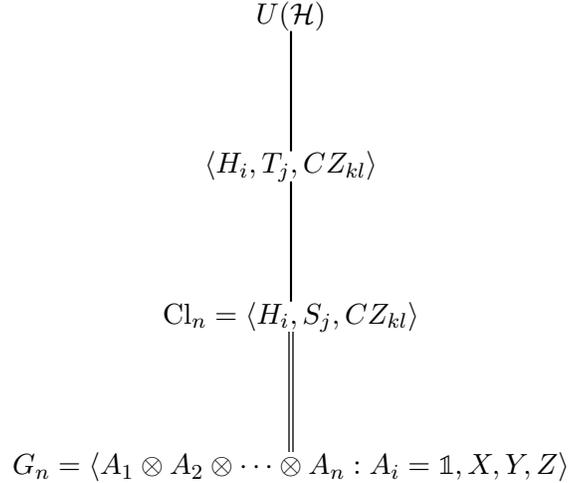

When designing a quantum circuit, that is, expressing a unitary $U$ as a product 
of finitely many elementary generators, one can in general only approximate a given 
unitary. The universality theorem of Kitaev--Solovay \cite{kitaev2002classical} 
states that the subgroup of the unitary group generated by $H_i$, $T_j$, and 
$\on{CZ}_{kl}$, where
\[ 
H=\tfrac{1}{\sqrt 2} \begin{pmatrix} 1 & 1 \\ 1 & -1\end{pmatrix}, \quad 
T=\begin{pmatrix} 1 & 0 \\ 0 & e^{i\pi/4}\end{pmatrix}, \quad
\on{CZ} = \begin{pmatrix} 
1&0&0&0\\
0&1&0&0\\			 
0&0&1&0\\
0&0&0&-1
\end{pmatrix},
\]
is dense in the unitary group. Subscripts specify the qubits on which the 
operators act; for example, the single-qubit unitary $H$, when written as $H_i$, 
denotes the $n$-qubit unitary that acts as the identity on all tensor factors 
except the $i$th, where it acts as $H$. 
In stark contrast, if $T_j$ is replaced with its square $S_j=T_j^2$, the subgroup 
generated by $H_i$, $S_j$, and $\on{CZ}_{kl}$ is a finite subgroup of the 
unitary group known as the Clifford group $\on{Cl}_n$. This group normalizes\footnote{More precisely, the normalizer of $G_n$ in the unitary 
group modulo scalar matrices is isomorphic to the Clifford group.} the 
Pauli group $G_n$ and specifies 
the transformations of stabilizer theory.

The Pauli group can be used to identify a distinguished set of quantum states. 
Elements of $G_n$ that square to the identity will be referred to as 
\emph{Pauli operators}. These operators have eigenvalues in $\set{\pm 1}$. 
Subgroups of $G_n$ that do not contain $-\one$ and are isomorphic to $\ZZ_2^k$ 
for some $1\leq k\leq n$ are called \emph{stabilizer subgroups}. For such a 
stabilizer subgroup $S=\Span{A_1,\dots,A_k}$, generated by $k$ pairwise 
commuting group elements, we can define the \emph{stabilizer projector}
\[
\Pi_{A_1\cdots A_k}^{s_1\cdots s_k} 
= \Pi_{A_1}^{s_1}\Pi_{A_2}^{s_2} \cdots \Pi_{A_k}^{s_k},
\] 
where $s_i\in \ZZ_2$ and $\Pi_A^s = (\one + (-1)^s A)/2$ for each unitary $A$. 
This is the projector onto the simultaneous eigenspace of the operators $A_i$ 
with eigenvalues $(-1)^{s_i}$. A quantum state is called a \emph{stabilizer 
state} if it is given by a stabilizer projector with $k=n$. For example, the 
projection operators onto the canonical basis of $(\CC^2)^{\otimes n}$ are 
stabilizer states corresponding to the subgroup generated by $Z_i$, where $Z_i$ 
is the tensor product of $Z$ on the $i$th factor and identities on the rest. 
Typically, a quantum computation is initiated by one of these computational 
basis states (see Figure~\ref{fig:circuit}).
Stabilizer projectors also give rise to a distinguished class of quantum 
measurements. A measurement associated to a single Pauli operator $A\in G_n$ is called a 
\emph{Pauli measurement}; its projectors are $\Pi_A^0$ and $\Pi_A^1$, projecting 
onto the $\pm 1$ eigenspaces of $A$. More generally, a measurement $\Pi$ is 
called a \emph{stabilizer measurement} if each projector $\Pi^s$ is a stabilizer 
projector. In quantum computation, the standard choice of measurements is the 
projectors onto the computational basis vectors. The \emph{stabilizer theory} is 
the subtheory of quantum mechanics consisting of stabilizer states, stabilizer 
transformations, and stabilizer measurements.

\subsection{Gottesmann--Knill theorem}
\label{sec:gottesmann knill theorem}

A quantum circuit is called a \emph{stabilizer circuit} if all of its components 
belong to the stabilizer theory. In Figure~\ref{fig:circuit}, this corresponds 
to taking $U$ in the Clifford group. The following result is known as the 
Gottesman--Knill (GK) theorem \cite{gottesman1999group22}.

\begin{thm}\label{thm:GK theorem}
Every stabilizer circuit can be efficiently simulated on a classical computer.  
\end{thm}

 In the remainder of this section we sketch the idea behind the proof of the 
theorem. As we will see, it relies heavily on group theory. The center of $G_n$ 
is the subgroup $\Span{i\one}$, which is isomorphic to the cyclic group of order 
$4$. The central quotient $E_n$ can be identified with $\ZZ_2^n\times \ZZ_2^n$. 
For an element $a\in E_n$ we write $(a^X,a^Z)$ for its two components. More 
explicitly, for each $a\in E_n$ we can define a Pauli operator
\[
T_{a} = i^{a^X\cdot a^Z} 
   X^{a^X_1}Z^{a^Z_1} \otimes \cdots \otimes X^{a^X_n}Z^{a^Z_n},
\]
where $a^X\cdot a^Z$ denotes the dot product. Any element of $G_n$ can then be 
written as $i^r T_{a}$, and under the quotient map this operator is mapped to 
$a$. In other words, Pauli operators can be represented by $2n$-bit strings, 
modulo the scalar $i^r$. Moreover, the commutation relation of Pauli operators 
is determined by a symplectic form $\omega$ on $E_n$:
\[
T_{a} T_{b} = (-1)^{\omega(a,b)} T_{b}T_{a},
\]
where $\omega(a,b) = a^X\cdot b^Z + b^X\cdot a^Z \pmod 2$. A subspace 
$I\subset E_n$ is called \emph{isotropic} if the restriction of $\omega$ to $I$ 
is zero. It is well known that there is a bijective correspondence between 
elementary abelian $2$-subgroups of $G_n$ and isotropic subspaces of $E_n$, 
given by the central quotient homomorphism. Once an isotropic subspace $I$ is 
fixed, there are $2^{\dim(I)}$ elementary abelian $2$-subgroups mapping to it 
under the quotient. For a proper parametrization of these groups, consider the 
relation
\[
T_{a} T_{b} = (-1)^{\beta(a,b)} T_{a+b}
\]
for pairs $(a,b)$ satisfying $\omega(a,b)=0$. That is, for commuting Pauli 
operators their product is given by the mod $2$ sum of their indices, up to a 
sign. Using this function $\beta$, one can parametrize stabilizer states. A more 
systematic description of $\beta$ can be obtained from the theory of group 
extensions. Observe that the extension class normally takes values in $\ZZ_4$, 
which is determined by the center, but when restricted to commuting pairs it 
takes values in $\ZZ_2$. A formal treatment requires chain complexes on the 
classifying space for commutativity \cite{Coho,OkaySheinbaum}.

\begin{defn}\label{def:value assignment}
A \emph{value assignment} on an isotropic subspace $I\subset E_n$ is a function 
$\gamma:I\to \ZZ_2$ such that 
\[
\gamma(a+b) = \gamma(a)+\gamma(b) + \beta(a,b)
\]
for all $a,b\in I$.
\end{defn}

As a consequence, a stabilizer state corresponding to 
$\Span{A_1,\cdots,A_n}\cong \ZZ_2^n$ can be specified by a sequence of $2n$ bit 
strings via the central quotient homomorphism. More precisely, stabilizer 
projectors can be parametrized by pairs $(I,\gamma)$, where $I$ is an isotropic 
subspace and $\gamma$ is a value assignment. The corresponding projector has the 
form
\[
\Pi_I^\gamma = \frac{1}{|I|} \sum_{a\in I} (-1)^{\gamma(a)} T_{a}.
\]
Stabilizer states correspond to the case $\dim(I)=n$. In stabilizer theory this 
data is organized into a matrix. For $I=\Span{a_1,\dots,a_n}$, the tableau is 
\[
\begin{pmatrix}
a_{1,1}^X & \cdots & a_{1,n}^X & \vline & b_{1,1}^Z & \cdots & b_{1,n}^Z & \vline & r_1 \\
\vdots    &        & \vdots    & \vline & \vdots    &        & \vdots    & \vline & \vdots \\
a_{n,1}^X & \cdots & a_{n,n}^X & \vline & b_{n,1}^Z & \cdots & b_{n,n}^Z & \vline & r_n
\end{pmatrix}
\]
where the $i$th row corresponds to the index $a_i$ of the operator 
$A_i=(-1)^{r_i}T_{a_i}$, with $r_i=\gamma(a_i)$.

A stabilizer circuit has the form shown in Figure~\ref{fig:circuit}, where the 
unitary $U$ belongs to the Clifford group. The first step in simulating such a 
circuit is to understand the action of $U$ in the binary picture. The algorithm 
relies on the following relation, see, e.g., \cite{raussendorf2023role}:
\[
U T_{a}U^\dagger = (-1)^{\phi_U(a)} T_{S(a)},
\]
where $S$ is a $2n\times 2n$ symplectic matrix acting on $E_n$ and 
$\phi:E_n\to \ZZ_2$ is a phase function. Ignoring the latter for simplicity, a 
Clifford unitary $U$ is essentially a symplectic transformation. Therefore, when 
the circuit is initialized in the state $\Pi^{0\cdots 0}$, the projector onto 
$\ket{0}\otimes \cdots \otimes \ket{0}$, the initial tableau is
\[
\begin{pmatrix}
1 &        &        & \vline &   &        &   & \vline &   \\
  & \ddots &        & \vline &   &        &   & \vline &   \\
  &        & 1      & \vline &   &        &   & \vline &   \\
\hline
  &        &        & \vline & 1 &        &   & \vline &   \\
  &        &        & \vline &   & \ddots &   & \vline &   \\
  &        &        & \vline &   &        & 1 & \vline &   
\end{pmatrix}
\]
where zeros are indicated by empty entries.
After the action of a Clifford unitary, each row is updated by applying the 
symplectic transformation $S$. A direct implementation requires $n$ 
matrix--vector multiplications, each costing $O(n^2)$ bit operations, leading to 
a 
total 
complexity\footnote{Aaronson--Gottesman show that each elementary 
Clifford gate ($H$, $S$, $\on{CNOT}$) can be simulated in $O(n)$ time {(per gate complexity)}. 
Since an arbitrary Clifford unitary can be decomposed into $O(n^2)$ elementary gates, 
this yields a total complexity of $O(n^3)$.} of $O(n^3)$.

The circuit terminates with measurements in the computational basis, which are 
stabilizer measurements. Suppose we measure a Pauli operator $B=T_{b}$. Writing 
$r\in \ZZ_2$ for the outcome, the measurement projectors are given by 
$\Pi_{b}^r$. Then, 
\begin{equation}\label{eq:stabilizer update}
\Pi_{b}^r \Pi_I^\gamma \Pi_{b}^r = 
\begin{cases}
\displaystyle\delta_{r,\gamma(b)}\Pi_I^\gamma & b\in I, \\[6pt]
\displaystyle\frac{|I(b)|}{|I|} \,\Pi_{I(b)}^{\gamma\ast r} & \text{otherwise,}
\end{cases}
\end{equation}
where $I(b)=I\cap \Span{b}^\perp$ and 
$J^\perp=\{a\in E_n : \omega(a,b)=0\;\;\forall b\in J\}$ for a subspace $J$. 
By dimensional considerations, if $b\notin I$ then $\dim(I(b))=n-1$ (so 
$|I(b)|/|I|=1/2$).
 We describe the updated 
value assignment $\gamma\ast r$ below. In other words, the formula says that when a 
Pauli measurement is performed on a stabilizer state:
\begin{itemize}
\item {Case I:} If the Pauli operator $B=T_{b}$ to be measured lies in the stabilizer 
group, then the outcome is $\gamma(b)$ with certainty, and the post-measurement 
state is the original stabilizer state.
\item {Case II:} Otherwise, the outcome is $0$ or $1$ with uniform probability. The 
post-measurement state is a new stabilizer state with stabilizer group generated 
by $(-1)^{\gamma\ast r(a)}T_{a}$ for $a\in I(b)$.
\end{itemize}
There is a simple way to compute $I(b)$. First, determine the first generator 
that does not commute with $b$, call it $j$. The generators of $I(b)$ are then 
\[
a_i' = 
\begin{cases}
a_i & i<j, \\
b   & i=j, \\
a_i+b & i>j.
\end{cases}
\] 
The value assignment $\gamma\ast r$ is determined on these generators by
\[
\gamma\ast r(a_i') =
\begin{cases}
\gamma(a_i) & i<j, \\
r & i=j, \\
\gamma(a_i)+r+\beta(a_i,b) & i>j.
\end{cases}
\]


{
In summary, the updates required during measurement are determined by linear 
operations on the tableau. Each update involves additions of two rows, 
costing $O(n)$ bit operations. {Distinguishing between Case~I and Case~II 
can be done by evaluating the symplectic form between $b$ and each generator, 
which requires $O(n^2)$ time.} The dominant cost arises from 
{determining deterministic outcomes in Case~I, which is equivalent to solving 
a system of linear equations over $\ZZ_2$ and can be performed by Gaussian 
elimination in $O(n^3)$ time.} Combining both parts of the algorithm, Clifford 
and measurement updates, the overall time complexity of stabilizer simulation is 
$O(n^3)$. The space complexity is $O(n^2)$, since the tableau has $n$ rows and 
$2n$ columns, together with an additional phase vector.
}

 In \cite{aaronson2004improved}, Aaronson and Gottesman provide an improvement of 
the last step, namely deciding whether the measurement operator belongs to the 
stabilizer of the state, by keeping track of the \emph{destabilizers} of the 
state. The destabilizers form a linearly independent set of vectors lying in the 
orthogonal complement of the isotropic subspace describing the stabilizer state. 
With this improvement, the tableau that describes a stabilizer state grows to 
contain $2n$ rows:
\begin{equation}\label{eq:improved tableau}
\renewcommand{\arraystretch}{1.4}
\begin{pmatrix}
D^X_{n\times n} & \vline & D^Z_{n\times n} & \vline & r_D \\
\hline
S^X_{n\times n} & \vline & S^Z_{n\times n} & \vline & r_S
\end{pmatrix}_{\Large 2n\times (2n+1)}.
\end{equation} 
The upper block corresponds to the destabilizer rows, and the lower block to the 
stabilizer rows.

\begin{rmk}\label{rem:qudit}
The stabilizer theory presented in this section for qubits generalizes in a 
straightforward way to qudits, specified by the $n$-qudit Hilbert space 
$(\CC^d)^{\otimes n}$. Concretely, this amounts to replacing the binary system 
$\ZZ_2=\{0,1\}$ with the $d$-level system 
$\ZZ_d=\{0,1,\dots,d-1\}$. The stabilizer theory is then built with respect to 
the $n$-qudit Pauli group generated by tensor products of the $d\times d$ Pauli 
matrices $X$ and $Z$, which act on computational basis vectors as
\[
X\ket{s} = \ket{s+1}, \qquad
Z\ket{s} = \mu^{s}\ket{s},
\]
where $\mu=e^{2\pi i/d}$. Although the constructions are formally similar, the 
algebraic properties of the Pauli group differ: when $d$ is odd the center of 
the group is isomorphic to $\ZZ_d$, whereas in the even case it is isomorphic to 
$\ZZ_{2d}$. This fundamental difference is reflected in the properties of 
classical simulation and in foundational considerations in quantum theory, such 
as the existence of hidden-variable models 
\cite{howard2014contextuality,delfosse2017equivalence}.
\end{rmk}

\section{Adaptive quantum computation} 
\label{sec:adaptive quantum computation}


To understand quantum computational advantage, it is necessary to look beyond the circuit model. Alternative models introduce an additional feature known as \emph{adaptivity}, whereby a quantum operation (typically a measurement) depends on the outcomes of previous operations. In practice, these operations are quantum measurements, which may be either destructive or non-destructive. A prominent example is the quantum computation with magic states (QCM) model, motivated by the efficient classical simulability of stabilizer circuits. In this section we develop a general computational framework based on adaptive instruments.

\subsection{Adaptive instruments}
\label{sec:adaptive instruments}

A uniform way of formulating quantum operations, including the three components (states, transformations, and measurements) is to use the notion of quantum channels and instruments.
A linear map $\phi:L(\hH)\to L(\kK)$ is called \emph{positive} if it maps positive semi-definite operators to positive semi-definite operators. The map $\phi$ is called \emph{completely positive} if $\phi\otimes \one_{L(\lL)}$ is positive for any Hilbert space $\lL$. A completely positive linear map is called a \emph{quantum channel} if it is trace-preserving.
We will write $\on{CP}(\hH,\kK)$ and $\on{C}(\hH,\kK)$ for completely positive linear maps and channels, respectively.

Quantum measurements and more general quantum operations can be formulated using instruments.
An \emph{instrument} with outcome set $\Sigma$ is a function $\Phi:\Sigma\to \on{CP}(V,W)$ with finite support 
such that $\sum_{s\in \Sigma} \Phi^s$ is a channel. Here we adopt the notation that $\Phi^s=\Phi(s)$.  We can regard an instrument as a quantum channel{
\begin{equation}\label{eq:Phi instrument}
\begin{aligned}
\Phi: L(\hH) &\to L(\kK\otimes \CC\Sigma)\\[1em]
\Phi(A) &= \sum_{s\in \Sigma} \Phi^s(A)\otimes  \ket{s}\bra{s}
\end{aligned}
\end{equation}}
where $\CC\Sigma$ denotes the free vector space with basis vectors $\{\ket{a}\colon a\in \Sigma\}$.

\begin{eg}
\label{ex:instruments}
Key examples of instruments are unitary operators, destructive and non-destructive measurements.
\begin{enumerate}
\item Quantum channels are examples of instruments where $\Sigma =\set{\ast}$. In particular, any unitary operator $U$ specifies a quantum channel 
\[
\Phi(A) = UAU^\dagger.
\]
\item A non-destructive measurement $\Pi:\Sigma \to \on{Proj}(\hH)$ specifies an instrument 
\[
\Phi^s(A) = \Pi^s A \Pi^s.
\]
\item A destructive measurement $\tilde \Pi:\Sigma \to \on{Proj}(\hH_1\otimes \hH_2)$ where $\tilde \Pi^s = \Pi^s \otimes \one_{\hH_2}$ can be regarded as an instrument 
\[
\Phi^s = \Tr_1 \circ \tilde\Phi^s\]
 where $\tilde \Phi^s$ is the instrument associated to the non-destructive measurement $\tilde \Pi$ and $\Tr_1$ is the partial trace over $\hH_1$.
\end{enumerate}
\end{eg}

\begin{defn}\label{def:adaptive instrument}  
An \emph{adaptive instrument} is a function
$$
\Gamma \times \Sigma \to \CP(\hH,\kK)
$$
that sends a pair $(a,s)$, where $a\in \Gamma$ is the input and $s\in \Sigma$ is the output, to a completely positive map $\Phi_a^s:L(\hH)\to {L}(\kK)$ such that $\sum_{s\in \Sigma} \Phi_a^s $ is a quantum channel for every input $a$.
\end{defn}

Note that we can regard an adaptive instrument as a linear map
{\begin{equation}\label{eq:Phi adaptive instrument}
\begin{aligned}
\Phi: L( \hH\otimes \CC\Gamma) &\to L(\kK\otimes \CC\Sigma)\\[1em]
\Phi( A \otimes \ket{a}\bra{b} ) &= \delta_{a,b} \sum_{s\in \Sigma} \Phi_a^s(A) \otimes \ket{s}\bra{s}.
\end{aligned}
\end{equation}}
We can depict adaptive instruments as boxes with two kinds of wires: (1) vertical classical wires, and (2) horizontal quantum wires:
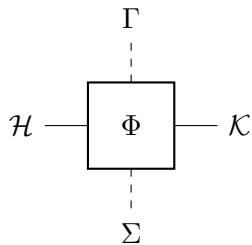
\begin{center}
\begin{tikzpicture}[x=0.5em,y=0.5em]

\Block{0}{0}{$\Phi$}{Phi}

\VIndraw{Phi}{1}{1}{3}
\node[above] at ($(Phi.north)+(0,3)$) {$\Gamma$};

\VOutdraw{Phi}{1}{1}{3}
\node[below] at ($(Phi.south)+(0,-3)$) {$\Sigma$};

\HIndraw{Phi}{1}{1}{3}
\node[left] at ($(Phi.west)+(-3,0)$) {$\cal{H}$};

\HOutdraw{Phi}{1}{1}{3}
\node[right] at ($(Phi.east)+(3,0)$) {$\cal{K}$};

\end{tikzpicture}
\end{center}
The horizontal wires indicate the quantum direction, that is, that direction has input the Hilbert space $\hH$ and output $\kK$. Whereas the vertical direction corresponds to the classical direction. This direction has input $\Gamma$ and output $\Sigma$. Moreover, such boxes can be composed in two directions, horizontally and vertically:
\begin{itemize}
\item The horizontal composition is given by
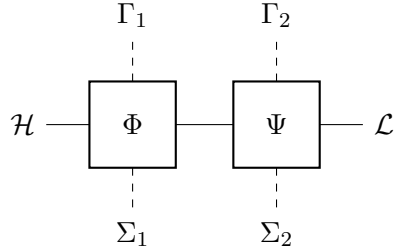
\begin{center}
\begin{tikzpicture}[x=0.5em,y=0.5em]
\def\xshift{30}
\def\hshift{10} 

\Block{\xshift}{0}{$\Phi$}{L}
\Block{\xshift+\hshift}{0}{$\Psi$}{R} 
 
\HDraw{L}{1}{1}{R}{1}{1}{0}{180}
\VIndraw{L}{1}{1}{3}
\VIndraw{R}{1}{1}{3}  
\HIndraw{L}{1}{1}{3} 
\HOutdraw{R}{1}{1}{3} 
\VOutdraw{L}{1}{1}{3}
\VOutdraw{R}{1}{1}{3}

\node[above] at ($(L.north)+(0,3)$) {$\Gamma_1$};
\node[above] at ($(R.north)+(0,3)$) {$\Gamma_2$};

\node[below] at ($(L.south)+(0,-3)$) {$\Sigma_1$};
\node[below] at ($(R.south)+(0,-3)$) {$\Sigma_2$};

\node[left]  at ($(L.west)+(-3,0)$) {$\cal{H}$};
\node[right] at ($(R.east)+(3,0)$) {$\cal{L}$};

\end{tikzpicture}
\end{center} 
where
\[
	(\Psi\circ \Phi)_{a_1,a_2}^{{s_1,s_2}}=\Psi_{a_2}^{{s_2}}\circ \Phi_{a_1}^{{s_1}} 
	\]

\item The vertical composition is given by 
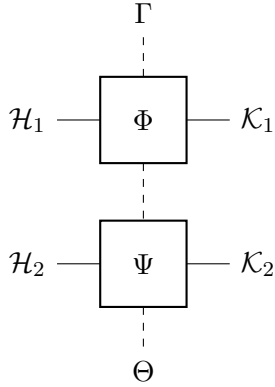
\begin{center}
\begin{tikzpicture}[x=0.5em,y=0.5em]
\def\xshift{30}
\def\vshift{10} 

\Block{\xshift}{0}{$\Phi$}{T}
\Block{\xshift}{-\vshift}{$\Psi$}{B} 
 
\VDraw{T}{1}{1}{B}{1}{1}{-90}{90}
\VIndraw{T}{1}{1}{3}
\HIndraw{T}{1}{1}{3}  
\HIndraw{B}{1}{1}{3} 
\HOutdraw{T}{1}{1}{3} 
\HOutdraw{B}{1}{1}{3}
\VOutdraw{B}{1}{1}{3}

\node[above] at ($(T.north)+(0,3)$) {$\Gamma$};          
\node[right] at ($(T.east)+(3,0)$) {${{\cal{K}}}_1$};            
\node[right] at ($(B.east)+(3,0)$) {${\cal{K}}_2$};            

\node[left] at ($(T.west)+(-3,0)$) {${\cal{H}}_1$};                 
\node[left] at ($(B.west)+(-3,0)$) {${\cal{H}}_2$};                 

\node[below] at ($(B.south)+(0,-3)$) {$\Theta$};               

\end{tikzpicture}
\end{center} 
where
 \[
(\Psi\bullet \Phi)_{a}^{{s}} = \sum_b \Phi_a^b\otimes \Psi_b^{{s}}
\]
\end{itemize}
More formally, the composition rules are compatible in the sense that instruments form a double category in the sense of \cite{grandis1999limits}. This double-categorical treatment of quantum computation will appear in \cite{double}. {The corresponding pictorial representation is particularly useful for visualizing adaptivity.}

\begin{defn}\label{def:adaptive composition}
Given adaptive instruments $\Phi:\Gamma\times \Sigma \to \CP(\hH,\kK)$ and $\Psi:\Sigma\times \Theta \to \CP(\kK,\lL)$, their \emph{adaptive composition} is the adaptive instrument
$$
\Psi\star \Phi: \Gamma\times (\Sigma\times \Theta) \to \CP(\hH,\lL)
$$
defined by
$$
(\Psi\star \Phi)_a^{s,r} =  \Psi_s^r \circ \Phi_a^s.
$$
\end{defn}

Adaptive composition cannot be decomposed into purely vertical or horizontal compositions, and it can be illustrated pictorially as follows:
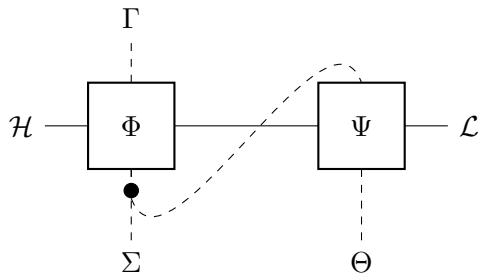
\begin{center}
\begin{tikzpicture}[x=0.5em,y=0.5em]
\def\xshift{30}
\def\hshift{10} 
\def\vshift{3}

\Block{\xshift}{0}{$\Phi$}{L}
\Block{\xshift+1.6*\hshift}{0}{$\Psi$}{R} 
\FCircle{\xshift}{-1.5*\vshift}{C}

\HDraw{L}{1}{1}{R}{1}{1}{0}{180}
\VDrawBC{L}{1}{1}{C}{1}{1}{-90}{90}
\VDrawCB{C}{1}{1}{R}{1}{1}{-70}{110}
\VIndraw{L}{1}{1}{3}   
\HIndraw{L}{1}{1}{3} 
\HOutdraw{R}{1}{1}{3}  
\VOutdraw{R}{1}{1}{5}
\VOutdraw{L}{1}{1}{5}

\node[above] at ($(L.north)+(0,3)$) {$\Gamma$};

\node[below] at ($(L.south)+(0,-5)$) {$\Sigma$};
\node[below] at ($(R.south)+(0,-5)$) {$\Theta$};

\node[left]  at ($(L.west)+(-3,0)$) {$\cal{H}$};
\node[right] at ($(R.east)+(3,0)$) {$\cal{L}$};

\end{tikzpicture}
\end{center}

\subsection{Adaptive computation}
\label{sec:adaptive computation}

\begin{defn}\label{def:adaptive computation}
An {\it adaptive quantum computation} consists of   
a sequence of adaptive instruments 
\[
\Phi_0,\Phi_1,\cdots,\Phi_N
\]
satisfying the following properties:
\begin{itemize}
\item each adaptive instrument has the form
$$
\Phi_i:\Gamma_{i}\times \Sigma_i \to \CP(\hH_{i-1},\hH_i),
$$ 
\item the initial instrument satisfies {$\Gamma_{0}=\Sigma_0=\set{\ast}$ and $\hH_0=\CC$},
\item for $i\geq 1$ each input set has the form $\Gamma_i = \Sigma_1 \times \Sigma_2 \times \cdots \times  \Sigma_{i-1}$.  
\end{itemize} 
 We say that the adaptive sequence $(\Phi_1,\cdots,\Phi_N)$ implements the quantum channel $\Phi\in C(\hH_0,\hH_N)$ if
$$
\Phi =  \sum_{s_0,\cdots,s_N}  \Phi^{{s_1}\cdots s_N}
$$
where
\[
\Phi^{{s_1}\cdots s_N}=(\Phi_N \star (  \cdots  \star {(}\Phi_1 \star \Phi_0)\cdots)^{{s_1}\cdots s_N}.
\]
\end{defn}

{Observe that the initial instrument $\Phi_0$ represents a state-preparation.} 
Thus $\rho=\Phi_0(1)$ is the input state of the quantum computation. The output state of the computation is given by 
\[
\rho^{s_1\cdots s_N}= \frac{\Phi^{s_1\cdots s_N}(1)}{\Tr(\Phi^{s_1\cdots s_N}(1))}
\]
when output $s_k$ is observed at instrument $\Phi_k$. Note that we have $s_0=\ast$ in this case. The probability of observing the output sequence $s_1\cdots s_N$ is given by the generalized version of the Born rule
\[
p^{s_1\cdots s_N} =  \Tr(\Phi^{s_1\cdots s_N}(1)).
\]

{
\begin{eg} 
In Section \ref{sec:quantum circuits}, we described how a computation can be modeled by a circuit consisting of a unitary operator $U$ followed by computational basis measurements. Such a process can be expressed as an adaptive computation involving a single instrument:
\[
\Phi^s(A) \;=\; \Pi^s U A U^\dagger \Pi^s,
\]
where $\Pi^s$ denotes the projection operator onto the computational basis vector labeled by $s \in \ZZ_2^n$.
\end{eg}
}

This model of computation, based on adaptive instruments, includes the quantum circuit model,  measurement-based quantum computation (MBQC) \cite{raussendorf2001one}, and quantum computation with magic states (QCM) \cite{bravyi2005universal}. We will describe models that represent the latter two models. Our approach in these constructions uses the diagramatic language of instruments. We specify a list of basic instruments from which any other can be constructed by composing them by gluing the classical/quantum wires.

\subsubsection{Pauli model}

The Pauli model is a particular case of QCM. We will be using the terminology introduced in Section \ref{sec:stabilizer theory}. The basic instruments are as follows:
\begin{itemize}
\item \emph{Resource state preparation} is done by the channel {$\Phi_P\in C(\CC,(\CC^2)^{\otimes n})$}, i.e., instrument with $\Gamma=\Sigma=\set{\ast}$, defined by
\[
\Phi_P(\alpha) = \alpha \rho  
\]  
{The model is universal with the choice of resource state $\rho_T = T \Pi_X^0 T^\dagger$.}
\item \emph{Pauli measurements} are implemented by the adaptive instrument $\Phi_{M}$, i.e.,  $\Gamma=\ZZ_2^k$ and $\Sigma=\ZZ_2$, defined by
\[
(\Phi_{M})_a^s(B) = \Pi_{A_a}^s B \Pi_{A_a}^s.
\] 
for some Pauli operators $A_a$.
\end{itemize}
{Then, an $n$-qubit computation in this model takes the form
\[
\Phi^{s_1\cdots s_N}
= \Phi_N \star \bigl(\cdots \star (\Phi_1 \star \Phi_0)\cdots\bigr)^{s_1\cdots s_N},
\]
where
\begin{itemize}
\item the initial state preparation is given by
\[
\Phi_0 = \underbrace{\Phi_P \bullet \cdots \bullet \Phi_P}_{n\ \text{times}},
\]
\item and for $k \ge 1$,
\[
\Phi_k = \Phi_M^k,
\]
with input set $\ZZ_2^{k-1}$ and output set $\ZZ_2$.
\end{itemize}}
For example, when $N=3$ the adaptive computation takes the form
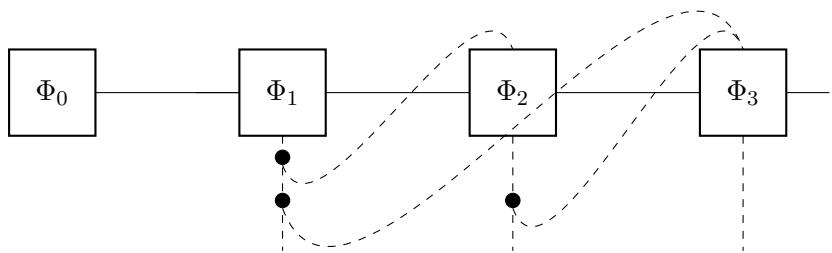
\begin{center}
\begin{tikzpicture}[x=0.5em,y=0.5em]
\def\xshift{30}
\def\hshift{10} 
\def\vshift{3}

\Block{\xshift-1.6*\hshift}{0}{$\Phi_0$}{LL}
\Block{\xshift}{0}{$\Phi_1$}{L}
\Block{\xshift+1.6*\hshift}{0}{$\Phi_2$}{M} 
\Block{\xshift+3.2*\hshift}{0}{$\Phi_3$}{R} 
\FCircle{\xshift}{-1.5*\vshift}{CL}
\FCircle{\xshift}{-2.5*\vshift}{CLD}
\FCircle{\xshift+1.6*\hshift}{-2.5*\vshift}{CMD}

\HDraw{LL}{1}{1}{L}{1}{1}{0}{180}
\HDraw{L}{1}{1}{M}{1}{1}{0}{180}
\HDraw{M}{1}{1}{R}{1}{1}{0}{180}


\VDrawCB{CL}{1}{1}{M}{1}{1}{-70}{110}
\VDrawCB{CMD}{1}{1}{R}{1}{1}{-70}{110}
\VDrawCB{CLD}{1}{1}{R}{1}{1}{-70}{110}

\HIndraw{L}{1}{1}{3} 
\HOutdraw{R}{1}{1}{3}  
\VOutdraw{R}{1}{1}{8}
\VOutdraw{L}{1}{1}{8}
\VOutdraw{M}{1}{1}{8} 
 



\end{tikzpicture}
\end{center}

\subsubsection{Local Pauli model}

We modify the previous model by restricting to single-qubit Pauli measurements and also making the measurements destructive. The latter means that the measured qubit is discarded. In effect, we implement this by taking a trace (see Example \ref{ex:instruments}). Then the resulting model fits into MBQC.
The basic instruments are as follows:
\begin{itemize}
\item \emph{Resource state preparation} is performed by the channel {$\Phi_P$, as in the Pauli model. The canonical choice ensuring universality is again the state $\rho_T$.}
\item \emph{Entanglement operation} is implemented by a unitary operator associated to a graph $G$ with edge set $E(G)$:
\[
\Phi_E(A) = U_G A U_G^\dagger
\]
where
\[
U_G = \prod_{\set{i,j}\in E(G)} \on{CZ}_{ij}.
\]
\item \emph{Local Pauli measurements} are implemented by the adaptive instrument $\Phi_{M}$, i.e.,  $\Gamma=\Sigma=\ZZ_2$, defined by
\[
(\Phi_{M})_a^s(B) = \on{Tr}_i(\Pi_{A_a}^s B \Pi_{A_a}^s)
\] 
where $A_a\in \set{X_i,Y_i,Z_i}$, and $\on{Tr}_i$ denotes partial trace on the $i$th qubit.
\item \emph{Classical feedforward} is implemented by the instrument $\Phi_L$ defined by
\[
(\Phi_L)_a^s(\alpha) = \delta_{s,f(a)} \alpha
\]
for some affine function $f:\ZZ_2^k\to \ZZ_2$.
\end{itemize}
{An $n$-qubit computation in this model takes the form
\[
\Phi^{s_1\cdots s_N}
= \Phi_N \star \bigl(\cdots \star (\Phi_1 \star \Phi_0)\cdots\bigr)^{s_1\cdots s_N},
\]
where
\begin{itemize}
\item the initial state preparation is given by
\[
\Phi_0 = \Phi_E\circ (\underbrace{\Phi_P \bullet \cdots \bullet \Phi_P}_{n\ \text{times}}),
\]
\item and for $k \ge 1$,
\[
\Phi_k = \Phi_{M}^{k} \bullet \Phi_L^{k-1},
\]
where $\Phi_L^k$ has input set $\ZZ_2^{k-1}$ and output set $\ZZ_2$.
\end{itemize}}
For example, when $N=3$ we have
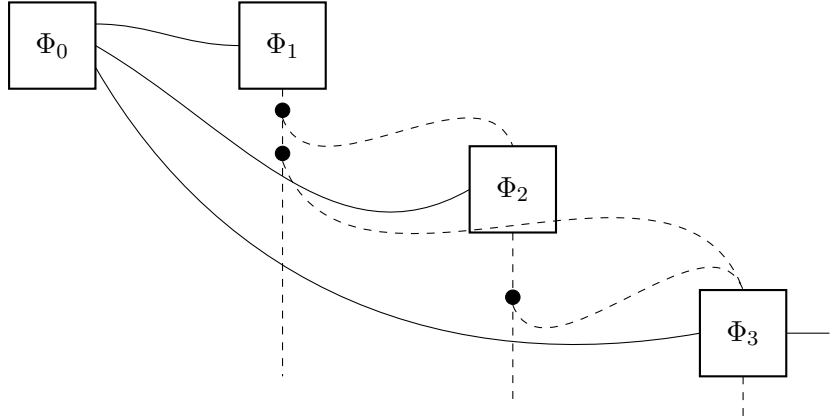
\begin{center}
\begin{tikzpicture}[x=0.5em,y=0.5em]
\def\xshift{30}
\def\hshift{10} 
\def\vshift{3}
\def\vdel{10}

\Block{\xshift-1.6*\hshift}{0}{$\Phi_0$}{LL}
\Block{\xshift}{0}{$\Phi_1$}{L}

\Block{\xshift+1.6*\hshift}{0-\vdel}{$\Phi_2$}{M} 

\Block{\xshift+3.2*\hshift}{0-2*\vdel}{$\Phi_3$}{R} 

\FCircle{\xshift}{-1.5*\vshift}{CL}
\FCircle{\xshift}{-2.5*\vshift}{CLD}

\FCircle{\xshift+1.6*\hshift}{-2.5*\vshift-\vdel}{CMD}

\HDraw{LL}{3}{1}{L}{1}{1}{0}{180}
\HDraw{LL}{3}{2}{M}{1}{1}{-30}{180+30}
\HDraw{LL}{3}{3}{R}{1}{1}{-60}{180+10} 


\VDrawCB{CL}{1}{1}{M}{1}{1}{-70}{110}
\VDrawCB{CMD}{1}{1}{R}{1}{1}{-70}{110}
\VDrawCB{CLD}{1}{1}{R}{1}{1}{-70}{110}

\HOutdraw{R}{1}{1}{3}  
\VOutdraw{R}{1}{1}{3}
\VOutdraw{L}{1}{1}{20}
\VOutdraw{M}{1}{1}{12} 
 



\end{tikzpicture}
\end{center} 


\section{{Polyhedral} classical simulation} 
\label{sec:polyhedral classical simulation}

Classical simulation methods that extend the Gottesman--Knill stabilizer simulation (described in Section \ref{sec:gottesmann knill theorem}) naturally exhibit a geometric flavor. In this section we introduce a general class of simulators, which we call {polyhedral classical simulators}, encompassing many of the simulation techniques developed in quantum computing. This framework is motivated by the simulation polytopes and the corresponding algorithms introduced in \cite{zurel2020hidden} and \cite{okay2024classical}.

Given a function $p : X \to D(Y)$, we denote by $\tilde{p} : D(X) \to D(Y)$ its convex extension,
\[
\tilde{p}\Bigl(\sum_x \lambda_x \delta^x\Bigr) = \sum_x \lambda_x p(x),
\]
{where $\delta^x$ denotes the Dirac (delta) distribution concentrated at $x$.}

\begin{defn}
Given sets $X$ and $Y$,
a {\it (probabilistic) update map} with outcome $\Sigma$ is a function
$$
q: X \to D(Y\times \Sigma).
$$
\end{defn}

We will write $q_x$ to denote the probability distribution corresponding to state $x\in X$.  
Let $\delta:\Sigma \to L(\CC\Sigma)$ denote the function {$s\mapsto \ket{s}\bra{s}$}.

\begin{defn} 
{A triple $(q;A,B)$ consisting of} an update map $q:X\to D(Y\times \Sigma)$  
{and}
functions
\begin{align*}
A:X &\to \Herm_1(\hH) \\
B:Y &\to \Herm_1(\kK)
\end{align*}
{\it simulates}  
{an}
instrument $\Phi$ if the following diagram commutes:
\[
\begin{tikzcd}[column sep=huge, row sep=large]
X \arrow[r,"q"] \arrow[d,"A"'] &  D(Y\times \Sigma) \arrow[d,"\widetilde{B\times \delta}"] \\
\Herm_1(\hH) \arrow[r,"\Phi"] & \Herm_1(\kK\otimes \CC\Sigma)  
\end{tikzcd}
\]
\end{defn}

{The map $\Phi$ is as given in Equation \ref{eq:Phi instrument}.}

\begin{defn}
\label{def:simulation}
We say that an instrument $\Phi$ {\it preserves}  a pair $(A,B)$ of functions
\begin{align*}
A:X &\to \Herm_1(\hH) \\
B:Y &\to \Herm_1(\kK)
\end{align*}
if for every $s\in \Sigma$, we have $\Tr(\Phi^s(A_x))\geq 0$ and the following two properties hold:
\begin{enumerate}
\item[(1)] $\Tr(\Phi^s(A_x))> 0$ implies
$$
\frac{\Phi^s(A_x)}{\Tr(\Phi^s(A_x))} \in \Conv(\set{B_y:y\in Y}).
$$
\item[(2)] $\Tr(\Phi^s(A_x))=0$ implies
$$
\Phi^s(A_x) =\zero.
$$
\end{enumerate} 
\end{defn}

In the definition we write $A_x$ and $B_y$ for the operators $A(x)$ and $B(y)$, respectively. 

{
\begin{eg}
Let $P$ and $Q$ denote the convex hulls of $\set{A_x : x \in X}$ and $\set{B_y : y \in Y}$, respectively.
A quantum channel $\Phi$ is said to preserve $(A,B)$ if and only if, for every $A \in P$, we have $\Phi(A) \in Q$.
In particular, when $\hH = \CC$, preservation under $\Phi$ reduces to the condition that the state $\rho = \Phi(1)$ lies in $Q$.
\end{eg}
}


\begin{thm}
\label{thm:correctness}
If $\Phi$ preserves $(A,B)$ then there exists an update map $q:X\to D(Y\times \Sigma)$ such that the triple $(q;A,B)$ simulates $\Phi$.
\end{thm}
\begin{proof}
If $\Phi$ preserves $(A,B)$ then we can define an  update map based on property (1) {of Definition \ref{def:simulation}}.
That is, if $\Tr(\Phi^s({A_x}))>0$ we have 
\begin{equation}\label{eq:update expansion}
\Phi^s(A_x) = \sum_{y} q_x(y,s) {B}_y
\end{equation}
for some probability distribution $q_x\in D(Y\times \Sigma)$. We define an update map $q:X\to D(Y\times \Sigma)$ by defining $q_x(y,s)$ using Equation (\ref{eq:update expansion}) if $\Tr(\Phi^s({A_x}))>0$, and zero otherwise.

To show that the resulting update map simulates the instrument,  
 we compute:
\begin{align*}
\widetilde{B\times \delta}\circ q(x) &= \widetilde{B\times \delta}\left( \sum_{y,s} q_x(y,s) {\delta^{y,s}} \right) \\
&= \sum_{s} \left( \sum_y q_x(y,s) {B}_y \right) \otimes \ket{s}\bra{s}  \\
&= {\sum_s} \Phi^s({A_x}) \otimes \ket{s}\bra{s}\\
& = \Phi\circ A(x)
\end{align*}
where we used Property (2) in line three. In details, by definition the sum $\sum_y q_x(y,s) {B_y}$ is equal to $\Phi^s({A_x})$ if $\Tr(\Phi^s({A_x}))>0$, and zero otherwise. When it is zero Property (2) implies that $\Phi^s({A_x})=\zero$, thus we have the desired equality in line three.
\end{proof}

Next, we turn to the adaptive case. An \emph{adaptive update map} is defined as a function
\[
q:X\times \Gamma \to D(Y\times \Sigma).
\]
The update map for $a\in \Gamma$ is denoted by $q_{a}$, and its value at $x\in X$ is denoted by {$q_{x,a}$}. We depict $q$ as a box: 
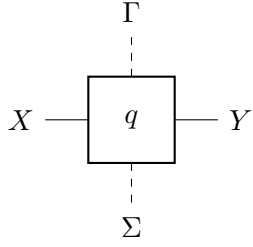
\begin{center}
\begin{tikzpicture}[x=0.5em,y=0.5em]

\Block{0}{0}{$q$}{Phi}

\VIndraw{Phi}{1}{1}{3}
\node[above] at ($(Phi.north)+(0,3)$) {$\Gamma$};

\VOutdraw{Phi}{1}{1}{3}
\node[below] at ($(Phi.south)+(0,-3)$) {$\Sigma$};

\HIndraw{Phi}{1}{1}{3}
\node[left] at ($(Phi.west)+(-3,0)$) {$X$};

\HOutdraw{Phi}{1}{1}{3}
\node[right] at ($(Phi.east)+(3,0)$) {$Y$};

\end{tikzpicture}
\end{center}
Given two adaptive update maps $q:X\times \Gamma\to D(Y\times \Sigma)$ and $p:Y\times \Sigma \to D(Z\times \Theta)$, we define their \emph{adaptive composition} to be the adaptive update map 
\[
p\star q : X\times \Gamma \to D(Z\times (\Sigma\times \Theta))
\] 
defined by
\[
(p\star q)_{x,a}(z,s,r) = \sum_y p_{y,s}(z,r) q_{x,a}(y,s).
\]
This composition can be depicted as
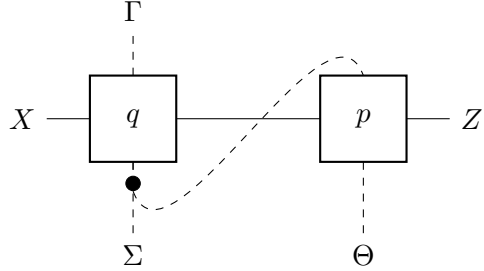
\begin{center}
\begin{tikzpicture}[x=0.5em,y=0.5em]
\def\xshift{30}
\def\hshift{10} 
\def\vshift{3}

\Block{\xshift}{0}{$q$}{L}
\Block{\xshift+1.6*\hshift}{0}{$p$}{R} 
\FCircle{\xshift}{-1.5*\vshift}{C}

\HDraw{L}{1}{1}{R}{1}{1}{0}{180}
\VDrawBC{L}{1}{1}{C}{1}{1}{-90}{90}
\VDrawCB{C}{1}{1}{R}{1}{1}{-70}{110}
\VIndraw{L}{1}{1}{3}   
\HIndraw{L}{1}{1}{3} 
\HOutdraw{R}{1}{1}{3}  
\VOutdraw{R}{1}{1}{5}
\VOutdraw{L}{1}{1}{5}

\node[above] at ($(L.north)+(0,3)$) {$\Gamma$};

\node[below] at ($(L.south)+(0,-5)$) {$\Sigma$};
\node[below] at ($(R.south)+(0,-5)$) {$\Theta$};

\node[left]  at ($(L.west)+(-3,0)$) {$X$};
\node[right] at ($(R.east)+(3,0)$) {$Z$};

\end{tikzpicture}
\end{center} 

We say that an adaptive instrument $\Phi$ \emph{preserves} $(A,B)$ if each instrument $\Phi_a$ preserves $(A,B)$.
Similarly, a triple $(q;A,B)$ \emph{simulates} an adaptive instrument $\Phi$ if each $q_a$ simulates $\Phi_a$.
This latter definition is equivalent to the commutativity of the following diagram: 
\[
\begin{tikzcd}[column sep=huge, row sep=large]
X\times \Gamma \arrow[r,"q"] \arrow[d,"A\times \delta"'] &  D(Y\times \Sigma) \arrow[d,"\widetilde{B\times \delta}"] \\
\Herm_1(\hH\otimes \CC\Gamma) \arrow[r,"\Phi"] & \Herm_1(\kK\otimes \CC\Sigma)  
\end{tikzcd}
\]

\begin{prop}\label{pro:composition} 
Given functions
\begin{align*}
A: X &\to \Herm_1(\hH) \\
B: Y &\to \Herm_1(\kK) \\
C: Z &\to \Herm_1(\lL) 
\end{align*}
and adaptive instruments 
\begin{align*}
\Phi: \Gamma\times  \Sigma & \to \CP(\hH,\kK)\\
\Psi: \Sigma\times  \Theta & \to \CP(\kK,\lL)
\end{align*}
where $\Phi$ and $\Psi$ preserve $(A,B)$ and $(B,C)$, respectively, the composite $\Psi\circ \Phi$ preserves $(A,C)$. 
Moreover, if $(q;A,B)$ and $(p;B,C)$ simulate $\Phi$ and $\Psi$, respectively, then $(p\star q;A,C)$ simulates the composite $\Psi		{\star}\Phi$, i.e., the following diagram commutes:
\[
\begin{tikzcd}[column sep=huge, row sep=large]
 X \times \Sigma \arrow[r,"p\star q"] \arrow[d,"A\times \delta"'] & D(Z\times (\Sigma \times \Theta)) \arrow[d,"\widetilde{C\times \delta}"]\\
\Herm_1(\hH\otimes \CC\Sigma) \arrow[r,"\Psi\star\Phi"] & \Herm_1(\lL\otimes \CC(\Sigma\times\Theta))
\end{tikzcd}
\]
\end{prop}
\begin{proof}
Given $A_x$, we have  
\[
(\Psi{\star}\Phi)_a^{s,r}(A_x)= \Psi_s^r(\Phi_a^s(A_x)) =
\begin{cases}
\sum_y q_{{x,a}}(y,s) \Psi_s^r({B_y})   & \Tr(\Phi_a^s(A_x))>0\\
\zero  & \Tr(\Phi_a^s(A_x))=0.
\end{cases}
\]
Then the condition $\Tr( (\Psi{\star}\Phi)_a^{s,r}(A_x) )\geq 0$ is satisfies. First, assume that the trace is equal to zero. Then either  $\Tr(\Phi_a^s(A_x))=0$ or $\Tr(\Phi_a^s(A_x))>0$ and $\Tr(\Psi_s^r({B_y}))=0$ for all $y$ such that $q_{{x,a}}(y,s)>0$. In both cases we see that $(\Psi{\star}\Phi)_a^{s,r}(A_x)=\zero$. Next, assume that the trace is positive. Then $\Tr(\Phi_a^s(A_x))>0$ and 
\begin{align*}
(\Psi\star\Phi)_a^{s,r}(A_x) &= \sum_y q_{{x,a}}(y,s) \Psi_s^r({B_y})  \\
&= \sum_{y:\,\Tr(\Psi_s^r({B_y}))>0} q_{{x,a}}(y,s) \sum_z p_{{y,s}}(z,r) {C_z} \\
&= \sum_z (p\star q)_{x,a}(z,s,r) {C_z}.
\end{align*}
This equation also implies the commutativity of the diagram concerning the simulation part of the statement.
\end{proof}

\begin{defn}\label{def:classical simulation}
A \emph{classical simulation algorithm} for an adaptive quantum computation 
$\Phi_0, \ldots, \Phi_N$ consists of a sequence of update maps 
\[
q_0, q_1, \ldots, q_N
\]
{and functions
\[
A_0,A_1, \ldots, A_N
\]}
satisfying the following properties:
\begin{itemize}
\item each adaptive update is a map of the form
\[
q_i : X_{i-1} \times \Gamma_{i} \to D(X_{i} \times \Sigma_i),
\]
\item the initial input satisfies $\Gamma_0 = \{\ast\}$,
\item for $i \ge 1$, the input set is given by 
\[
\Gamma_i = \Sigma_1 \times \Sigma_2 \times \cdots \times \Sigma_{i-1},
\]
\item each function $A_i$ is of the form
\[
A_i : X_i \to \Herm_1(\hH_i),
\]
\item {each triple $(q_i; A_{i-1}, A_i)$ simulates $\Phi_i$.}
\end{itemize}
\end{defn}

{Observe that $A_0:X_0\to \Herm_1(\CC)=\set{1}$ is the constant function. Proposition \ref{pro:composition} implies that the composite
\[
q_N\star (\cdots\star(q_1\star q_0)\cdots)
\]
simulates the instrument $\Phi:\Sigma_1\times \cdots \times \Sigma_N\to \on{CP}(\hH_0,\hH_N)$ defined by the composite
\[
\Phi= \Phi_N\star (\cdots\star(\Phi_1\star \Phi_0)\cdots). 
\]
In practice, polyhedral simulators arise directly from Theorem~\ref{thm:correctness} via the preservation property.

\begin{cor}\label{cor:preservation implies algorithm}
Let $(\Phi_0, \ldots, \Phi_N)$ be an adaptive quantum computation, and let 
$A_i : X_i \to \Herm(\hH_i)$ be a collection of functions such that each 
$\Phi_i$ preserves $(A_{i-1}, A_i)$. Then there exists a classical simulation 
algorithm simulating the computation.
\end{cor}
}

{As seen in the Pauli and local Pauli cases, universal models of adaptive quantum computation can be constructed from a small set of basic adaptive instruments. Hence, it suffices to verify the preservation property for these basic instruments.}

We begin with two basic simulators. The stabilizer formalism introduced in Section \ref{sec:stabilizer theory} for qubits extends naturally to qudits; see Remark \ref{rem:qudit}. {Next, we consider two simulators that can be used to simulate an adaptive quantum computation in the Pauli model:
\[
(\Phi_0,\Phi_1,\cdots,\Phi_N)
\]
where $\Phi_0$ is the state preparation and for $i\geq 1$, each $\Phi_i$ is a non-destructive Pauli measurement. For $n$ qudits the Hilbert space is $\hH_i = (\CC^d)^{\otimes n}$.
}

\paragraph{Stabilizer simulator.} 
\begin{itemize}
\item $X_i$ consists of pairs $(I,\gamma)$ where $I$ is a maximal isotropic subspace of $E_n=\ZZ_d^{2n}$ and $\gamma:I\to \ZZ_d$ is a value assignment.
\item $A_i:X_i\to \Herm_1(\hH_i)$ sends $(I,\gamma)$ to the projector 
\[
\Pi_I^\gamma = \frac{1}{|I|} \sum_{a\in I} \mu^{\gamma(a)} T_a,
\]
where $\mu = e^{2\pi i/d}$.
\item Update maps are obtained from the generalization of Equation~(\ref{eq:stabilizer update}) to qudits, 
as given in \cite[Eq.~19]{zurel2021hidden}.
\end{itemize}
The simulation polytope in this case is the stabilizer polytope $\on{SP}_n$ given by the convex hull of the stabilizer states.
 
\paragraph{Wigner simulator.} {In this case $d$ is odd.}
\begin{itemize}
\item $X_i$ consists of value assignments $\gamma:E_n\to \ZZ_d$.  
\item $A_i:X_i\to \Herm_1(\hH_i)$ sends $\gamma$ to the phase-point operator 
\[
A^\gamma = \frac{1}{d^n} \sum_{a\in E_n} \mu^{\gamma(a)} T_a.
\]
\item Update maps are described in \cite[Lemma~5]{zurel2024efficient}.
\end{itemize}
The simulation polytope is the Wigner polytope $\on{WP}_n$ given by the convex hull of $A^\gamma$'s. 
 
The stabilizer simulator is essentially the Gottesman--Knill method \cite{gottesman1999group22,aaronson2004improved} expressed in the polyhedral framework. The Wigner simulator, introduced in \cite{veitch2012negative}, is historically important as a successor to the stabilizer simulator. It applies only to qudits of odd local dimension, and its failure for qubits has motivated the development of the simulators that follow. See Figure \ref{fig:polytope-hierarchy} for the relation to other polyhedral simulators. The remaining polytopes appearing in that diagram will be introduced in the course of this section.

\subsection{Extended stabilizer simulator}

Next, we describe a natural extension of the stabilizer theory {and the corresponding simulator for the Pauli model}. The key step is the following generalization of stabilizer state projectors.

\begin{defn}\label{def:cnc}
A subset $\Omega \subset E_n$ is called 
\begin{itemize}
\item \emph{closed} if $a,b\in \Omega$ with $\omega(a,b)=0$ implies $a+b\in \Omega$,
\item \emph{non-contextual} if it admits a value assignment $\gamma:\Omega\to \ZZ_2$, i.e., a function satisfying
\[
\gamma(a+b) = \gamma(a)+\gamma(b)+\beta(a,b)
\]
for all $a,b\in \Omega$ such that $\omega(a,b)=0$.
\end{itemize}
A \emph{closed non-contextual (CNC) set} is a subset that is both closed and non-contextual.
\end{defn}

Every isotropic subspace $I$ is in particular a CNC set. 
Recall that, together with a value assignment $\gamma$, the pair $(I,\gamma)$ specifies a stabilizer projector. 
In the same way, to a CNC set we can associate the operator
\[
A_\Omega^\gamma = \frac{1}{2^n} \sum_{a\in \Omega} (-1)^{\gamma(a)} T_a,
\]
called a \emph{closed non-contextual (CNC) operator}. The set of maximal closed non-contextual operators forms the \emph{phase space}. 
This notion extends naturally to qudits \cite{zurel2024efficient}.

\paragraph{Phase space (CNC) simulator.} This simulation method was first introduced in \cite{raussendorf2020phase}. 
The corresponding tableau algorithm, which extends the stabilizer tableau, is developed recently in \cite{ipek2025phase} and implemented in \cite{CNCSim}. 
Here $\hH_i = (\CC^2)^{\otimes n}$ denotes the $n$-qubit Hilbert space.
\begin{itemize}
\item $X_i$ consists of pairs $(\Omega,\gamma)$ where $\Omega \subset E_n$ is a maximal CNC set and $\gamma:\Omega\to \ZZ_2$ is a value assignment. 
\item $A_i:X_i\to \Herm_1(\hH_i)$ sends $(\Omega,\gamma)$ to the operator $A_\Omega^\gamma$. 
\item Update maps are described in \cite{ipek2025phase}. 
\end{itemize}
The convex hull of the CNC operators is called the CNC polytope, denoted $\on{CP}_n$, which serves as the simulation polytope in this case.

The comparison with the stabilizer tableau method is as follows:  
a maximal closed non-contextual operator $A_\Omega^\gamma$ is represented by a binary tableau
\[
\renewcommand{\arraystretch}{1.4}
\begin{pmatrix}
D^X_{(n-m)\times n} & \vline & D^Z_{(n-m)\times n} & \vline & r_D \\
\hline
S^X_{(n-m)\times n} & \vline & S^Z_{(n-m)\times n} & \vline & r_S \\
\hline
J^X_{2m\times n} & \vline & J^Z_{2m\times n} & \vline & r_J
\end{pmatrix}_{\Large 2n\times (2n+1)},
\]
consisting of destabilizer, stabilizer, and \emph{Jordan--Wigner} parts. 
The last component introduces an additional symplectic structure in the phase space simulator. 
When $m=0$, this tableau reduces to the stabilizer tableau of Equation~(\ref{eq:improved tableau}).

Simulation of Clifford unitaries proceeds in the same way as in the stabilizer case. 
The main novelty arises in the measurement step, where the two cases of the stabilizer simulation refine into four:  
\begin{itemize}
\item Case I: analogous to the deterministic case in the stabilizer tableau,  
\item Case II/III: new (and more involved) probabilistic updates,  
\item Case IV: analogous to the probabilistic case in the stabilizer tableau.  
\end{itemize}
It is shown in \cite{ipek2025phase} that the complexity of this simulation matches 
that of the Aaronson--Gottesman improved algorithm: $O(n^2)$ bits of memory for 
the tableau, $O(n^3)$ time complexity for Clifford unitaries, and 
$O(n^2)$ for Pauli measurements.

\subsection{Universal samplers}

The polyhedral simulators—stabilizer, Wigner, and phase space—efficiently simulate a quantum computation whenever the initial quantum state lies in the corresponding polytope. These polytopes are:  
\begin{itemize}
\item the \emph{stabilizer polytope} $\on{SP}_n$, consisting of convex combinations of $n$-qubit stabilizer states,  
\item the \emph{Wigner polytope} $\on{WP}_n$, consisting of convex combinations of $n$-qubit phase-point operators,  
\item the \emph{CNC polytope} $\on{CP}_n$, consisting of convex combinations of closed non-contextual operators.  
\end{itemize}
When the initial state does not belong to the relevant polytope, one can still perform classical simulation by \emph{estimating} the Born-rule probabilities, following the method of Pashayan et al.~\cite{pashayan2015estimating}. For instance, the implemented phase space simulator \cite{CNCSim} operates on this principle. Such a simulator will be referred to as \emph{efficient estimator}.

The basic idea is that when the initial state lies outside a given polytope, it can no longer be expressed as a convex combination of the polytope’s vertices. Hence, the corresponding probability distribution cannot be obtained directly. Instead, the initial state can be written as a quasi-probabilistic mixture of the vertices:
\[
\rho = \sum_{\alpha} r_\alpha A_\alpha,
\]
where $r_\alpha \in \RR$ and $\sum_\alpha r_\alpha = 1$. One can then use the probability distribution 
\[
p_\alpha = \frac{|r_\alpha|}{\sum_\alpha |r_\alpha|}
\]
to initiate the simulation. The overhead of this simulation scales quadratically in the quantity known as the \emph{robustness}:
\[
\mathfrak{R}(\rho) = \min \sum_\alpha |r_\alpha|,
\]
where the minimum is taken over all decompositions $\rho = \sum_\alpha r_\alpha A_\alpha$. By definition, this quantity depends on the choice of polytope used for the classical simulation. 
This leads to the natural question: Are there polytopes with $\mathfrak{R}(\rho)=1$, equivalently, polytopes that contain every quantum state? Two such polytopes have been described recently. We refer to the corresponding classical simulator as a \emph{universal sampler}.  

 \paragraph{Pauli simulator.}  
The first example of a polytope containing all quantum states was introduced in \cite{zurel2020hidden} for qubits and later generalized to qudits in \cite{zurel2021hidden}. 

\begin{defn}
The \emph{$n$-qubit Pauli polytope}, denoted by $\on{P}_n$, is defined as the dual\footnote{By duality we mean polar duality up to reflection, which coincides with the notion of duality used in generalized probabilistic theories~\cite{plavala2023general}.}
of the stabilizer polytope:
\[
\on{P}_n = \set{ A \in \Herm_1\!\big((\CC^2)^{\otimes n}\big) \;\colon\; \Tr(A\Pi)\geq 0 \;\; \text{for every stabilizer state } \Pi }.
\]
\end{defn}

By construction, every $n$-qubit quantum state lies in $\on{P}_n$. {This polytope serves as a universal sampler for the Pauli model, capable of simulating any adaptive quantum computation expressed in this model.} 
The associated simulation is described as follows: {For $i=1,\cdots,N$,} $\hH_i = (\CC^2)^{\otimes n}$ is the $n$-qubit Hilbert space.
\begin{itemize} 
\item {$X_i$ is a set of labels for the vertices of $\on{P}_n$.}
\item $A_i:X_i\to \Herm_1(\hH_i)$ is the inclusion map.  
\item Update maps are not known in general.  
\end{itemize}
{The classical simulation algorithm can be obtained as a result of Corollary \ref{cor:preservation implies algorithm} once we show that the basic instruments in the Pauli model preserves the $A_i$'s. The following result is preserved in \cite{zurel2020hidden}.

\begin{thm}
Every Pauli measurement $\Phi_M$ preserves $(A,A)$, where $A : X \to \Herm_1((\CC^2)^{\otimes n})$ denotes the inclusion of the label set corresponding to the vertices of $\on{P}_n$.
\end{thm}
}

The CNC polytope is contained within the Pauli polytope. Moreover, the maximal CNC operators---the vertices of $\on{CP}_n$---are also vertices of the Pauli polytope. In this special case, the description of the vertices and their update rules is therefore known. Identifying new vertices of the Pauli polytope and determining their updates under Pauli measurements remains an active area of research.

\paragraph{Local Pauli simulator} 
Another polytope that contains all $n$-qubit quantum states was recently introduced in \cite{okay2024classical}. 
This polytope is defined as the dual of the \emph{local stabilizer polytope}. 
A stabilizer state is called \emph{local} if its stabilizer group is generated by pairwise commuting single-qubit Pauli operators. {More precisely, such a stabilizer group can be written as $\Span{A_1,\cdots,A_n}$ where each $A_i\in \set{X_i,Y_i,Z_i}$.}

\begin{defn}
The \emph{$n$-qubit local Pauli polytope}, denoted by $\on{LP}_n$, is defined as the dual of the local stabilizer polytope:
\[
\on{LP}_n = \set{A\in \Herm_1\!\big((\CC^2)^{\otimes n}\big) \;\colon\; \Tr(A\Pi)\geq 0 \;\; \text{for every local stabilizer state } \Pi }.
\]
\end{defn}

By definition, $\on{P}_n$ is contained in $\on{LP}_n$. {Using this polytope, we can simulate any adaptive quantum computation in the local Pauli model, thereby obtaining a universal sampler for this computational model.}
The corresponding simulation algorithm is as follows:  
{For $i=1,\cdots,N$, $\hH_i = (\CC^2)^{\otimes n-i+1}$, the $n-i+1$-qubit} Hilbert space.
\begin{itemize} 
\item {$X_i$ is a set of labels for the vertices of $\on{LP}_{n-i+1}$.}
\item $A_i:X_i\to \Herm_1(\hH_i)$ is the inclusion map.  
\item Update maps are not known in general.  
\end{itemize}
{Observe that the number of qubits decreases as a result of the use of destructive measurements. The classical simulation algorithm is obtained from Corollary \ref{cor:preservation implies algorithm} and the following result proved in \cite{okay2024classical}.

\begin{thm}
Every local Pauli measurement $\Phi_M$ preserves $(A_0,A_1)$, where $A_i : X \to \Herm_1((\CC^2)^{\otimes n-i})$ denotes the inclusion of the label set corresponding to the vertices of $\on{LP}_{n-i}$.
\end{thm}
}

The local Pauli polytope consists of operators that are local versions of CNC operators, referred to as \emph{locally closed operators}. Two classes of vertices have been identified in the corresponding polytope, the \emph{locally closed polytope} $\on{LC}_n$:
 the \emph{deterministic} vertices and the \emph{max-weight} vertices.
The corresponding simulation polytopes are denoted by $\on{DP}_n$ and $\on{MP}_n$. See Figure \ref{fig:polytope-hierarchy} for the relationships among the polytopes introduced so far.
 In \cite{okay2024classical} it is shown that the robustness measure for simulations in the local Pauli simulator is lower than that for simulations based on the Pauli simulator. This suggests that the local Pauli simulator provides a more efficient method of classical simulation than the Pauli simulator.


	\bibliography{bib.bib}
\bibliographystyle{ieeetr}
	
\end{document}